\documentclass[a4paper,english]{article}
\usepackage{microtype}

\usepackage[utf8]{inputenc}
\usepackage{authblk}
\usepackage{booktabs,subcaption,dcolumn}
\usepackage{hyperref}
\usepackage{amsthm}
\usepackage{amsfonts}
\usepackage{fullpage}
\usepackage{authblk}
\usepackage{newtxtext,newtxmath}
\usepackage{tikz}
\usetikzlibrary{arrows}
\usetikzlibrary{patterns}
\usepackage{todonotes}
\usepackage{caption}
\usepackage{xspace}
\usepackage[linesnumbered, ruled, vlined]{algorithm2e}
\usepackage{standalone}
\usepackage{csquotes}
\usepackage[toc,page]{appendix} 

\theoremstyle{definition}
\newtheorem{theorem}{Theorem}

\newtheorem{lemma}[theorem]{Lemma}
\newtheorem{definition}[theorem]{Definition}
\newtheorem{corollary}[theorem]{Corollary}

\newcommand{\N}{\mathbb{N}}

\newcommand{\cktPlus}[1]{\ensuremath{\textsc{ADD}_{{#1}}}}  
\newcommand{\cktDifference}[1]{\ensuremath{\textsc{SUB}_{{#1}}}}  
\newcommand{\cktSum}[2]{\ensuremath{\textsc{SUM}_{{#1}, {#2}}}}  
\newcommand{\cktOnes}[1]{\ensuremath{\textsc{ONES}_{{#1}}}}  
\newcommand{\cktOnesDeep}[1]{\ensuremath{\textsc{ONES'}_{{#1}}}}  
\newcommand{\cktSwitch}[1]{\ensuremath{\textsc{SWITCH}_{{#1}}}}  
\newcommand{\cktAKS}[2]{\ensuremath{\textsc{AKS}_{{#1}, {#2}}}}  
\newcommand{\cktAKSPartial}[3]{\ensuremath{\textsc{PARTIAL\_AKS}_{{#1}, {#2}, {#3}}}}  
\newcommand{\cktPSort}[3]{\ensuremath{\textsc{PIPPENGER\_SORT}_{{#1}, {#2}, {#3}}}}  
\newcommand{\cktRSort}[3]{\ensuremath{\textsc{ITERATIVE\_SORT}_{{#1}, {#2}, {#3}}}}  
\newcommand{\cktCount}[2]{\ensuremath{\textsc{COUNT}_{{#1}, {#2}}}}  
\newcommand{\cktDecompress}[2]{\ensuremath{\textsc{DECOMPRESS}_{{#1}, {#2}}}}  
\newcommand{\cktFastCount}[2]{\ensuremath{\textsc{FAST\_COUNT}_{{#1}, {#2}}}}  
\newcommand{\cktFastDecompress}[2]{\ensuremath{\textsc{FAST\_DECOMPRESS}_{{#1}, {#2}}}}  
\newcommand{\cktRoute}[1]{\ensuremath{\textsc{ROUTE}_{{#1}}}}  
\newcommand{\cktSort}[3]{\ensuremath{\textsc{SORT}_{{#1}, {#2}, {#3}}}}  

\newcommand{\fromto}[2]{\ensuremath{ \colon \left\{ 0,1 \right\}^{#1} \rightarrow \left\{ 0,1 \right\}^{#2} }}
\newcommand{\abs}[1]{\ensuremath{\left| #1 \right|}}

\title{Sorting Short Integers\footnote{Michal Koucký and Karel Král were supported by the Grant Agency of the Czech Republic under the grant agreement no. 19-27871X. Karel Král was also partially supported by the Charles University project SVV–2020–260578.}}

\author[1]{Michal Koucký}
\author[1]{Karel Král}
\affil[1]{Computer Science Institute, Charles University, Prague, Czech Republic\protect\\
\texttt{\{koucky, kralka\}@iuuk.mff.cuni.cz}}

\begin{document}

\maketitle

\begin{abstract}
We build boolean circuits of size $\mathcal{O}(nm^2)$ and depth $\mathcal{O}(\log(n) + m \log(m))$ for sorting $n$ integers each of $m$-bits. 
We build also circuits that sort $n$ integers each of $m$-bits according to their first $k$ bits that are of size $\mathcal{O}(nmk (1 + \log^*(n) - \log^*(m)))$ and depth $\mathcal{O}(\log^{3}(n))$. This improves on the results of Asharov~et~al.~\cite{asharov2021sorting} and resolves some of their open questions.
\end{abstract}

\section{Introduction}
\label{sec:introduction}

Sorting undoubtedly plays a central role in computer science.
Great many problems can be solved using sorting as a subcomponent.
There are many practical variants of sorting based either on what we sort (integers, rational numbers, strings, etc.) or how we sort (in parallel,  in distributed fashion, in external memory, etc.). Despite lots of research there are still many basic questions about sorting unanswered.

The classical comparison based sorting takes time $\mathcal{O}(n \log(n))$ when sorting $n$ integers. Well known lower bound postulates that this is optimal for comparison based sorting. However, this is a great over-simplification and the picture is much more nuanced: sorting integers from a domain of size $M$ can be done using binary search trees in time $\mathcal{O}(n \log|M|)$, thus sorting for example $m$-bit integers only needs $\mathcal{O}(n m)$ comparisons. Such an algorithm can be implemented on a pointer machine, for example. In the RAM model, with the word size $m$ we can sort even faster: 
When $m=O(\log(n))$ one can sort in time $\mathcal{O}(n)$ using radix sort, and when $m=\Omega(\log^3(n))$ one can also sort in linear time using the algorithm of Andersson~\cite{andersson1998sorting}. When $m=O(\log^3(m))$ one can sort in expected time $\mathcal{O}\left(n \sqrt{\log \frac{m}{\log(n)}}\right)$ and linear space using the algorithm of Han and Thorup~\cite{han2002sorting}. It is an easy exercise to design Turing machines that sort $m$-bit integers in time $\mathcal{O}(n m^2)$.

In many cryptographic applications there is an interest in oblivious algorithms, algorithms in which the sequence of the operations is independent of the processed data. Sorting plays an important role in construction of oblivious RAM. An oblivious comparison based parallel model of computation intended for sorting are \emph{sorting networks}. Numbers in a sorting network are thought of as signals which can only be compared.
The seminal paper by Ajtai, Komlós, and Szemerédi~\cite{aks_1983sorting} gives an asymptotically optimal sorting network of logarithmic depth and thus having $\mathcal{O}(n \log(n))$ comparators matching the comparison based lower bound.
The AKS network has immense applications in theoretical computer science, and we use it in this paper, too.

Another oblivious model of computation heavily used throughout theoretical computer science are boolean circuits. One can turn the AKS sorting network into a circuit of size $\mathcal{O}(n m \log(n))$ and depth $\mathcal{O}(\log(m) \log(n))$ (see Section~\ref{sec:sorting}). However, when building boolean circuits for sorting it is not clear whether one can take any advantage of some of the faster algorithms for RAM or Turing machines as simulating random access memory or Turing machine tapes by circuits requires substantial overhead. Asharov~et~al.~\cite{asharov2021sorting} asked the question whether one can sort $m$-bit integers in time $o(n m\log(n))$ when $m=o(\log(n))$. They provide an answer to this question by constructing circuits for sorting $m$-bit integers of size $\mathcal{O}(nm^2 (1 + \log^*(n) - \log^*(m))^{2 + \varepsilon})$ and polynomial depth, for any $\varepsilon > 0$. We improve their results: We build boolean circuits for sorting $m$-bit integers of size $\mathcal{O}(nm^2)$ and depth $\mathcal{O}(\log(n) + m \log(m))$. Pending some unexpected breakthrough this size seems optimal. The depth is provably optimal whenever $m = O(\log(n)/\log \log(n))$.

Asharov~et~al.~\cite{asharov2021sorting} solve even a more general problem as their circuits partially sort $n$ numbers each of $m$ bits by their first $k$ bits using a circuit of size $\mathcal{O}(nmk (1 + \log^*(n) - \log^*(m))^{2 + \varepsilon})$. We improve on this result as well by presenting circuits that sort $m$-bit integers according to their first $k$ bits of size $\mathcal{O}(nmk (1 + \log^*(n) - \log^*(m)))$ and depth $\mathcal{O}(\log^{3}(n))$. Our small circuits of poly-logarithmic depth answer some of the open questions of Asharov~et~al.~\cite{asharov2021sorting}.
In a work subsequent to ours, Lin and Shi~\cite{lin2021optimal} get circuits of depth $\mathcal{O}(\log(n) + \log(k))$ and size $\mathcal{O}(nkm \cdot \text{poly}(\log^*(n) - \log^*(m)))$ whenever $n > 2^{4k+7}$.
They use substantially different approach.
We state our results in the next section.

\subsection{Our Results}
\label{sec:our_results}

We provide a family of boolean circuits that sort $m$-bit strings. Our circuits are smaller than the circuits directly derived from the AKS sorting network, and they improve on the result of Asharov~et~al.~\cite{asharov2021sorting}. Our circuits achieve optimal logarithmic depth whenever $m \log(m) \leq \log(n)$. Pending some unexpected breakthrough, their size seems also optimal.

\begin{theorem}
     For any integers $n,m\ge 1$ there is a size $\mathcal{O}(nm^2)$ and depth $\mathcal{O}(\log(n) + m \log(m))$ circuit
that sorts $n$ integers of $m$ bits each.
	\label{thm:main}
\end{theorem}

For $m\ge \Omega(\log(n))$, the existence of such a circuit directly follows from AKS sorting networks. Our contribution is the construction of such circuits for $m \le o(\log(n))$. Our construction also uses a sorting network as a building block.
We use the AKS sorting network as one of our primitives but in principle, we could use any sorting network or sorting circuit.
In particular, we could use any circuit sorting $n$ numbers of $\log(n)$ bits each  in our construction.
Any improvement of asymptotic complexity of sorting of $\log(n)$-bit numbers would give us improved complexity of sorting short numbers.

The main idea behind our construction is to \emph{compress} the input by computing the number of occurrences of each $m$-bit integer. 
This gives a vector of $2^m$ integers, each of size $\mathcal{O}(\log(n))$. Decompressing this vector back gives the sorted input.
Combining the counting and decompressing circuit gives us a circuit that sorts.
The main technical lemma is our counting circuit which is of independent interest. 

\begin{lemma}
        For any integers $n,m\ge 1$ where $m \leq \log(n) / 10$ there is a circuit $$\cktFastCount{n}{m} \fromto{n \cdot m}{\lceil 1+\log(n) \rceil 2^{m}}$$ which given a sequence of $n$ strings of $m$ bits each outputs the number of occurrences of each possible $m$-bit string among the inputs,
that is for input $x_1, x_2, \ldots, x_n \in \left\{ 0,1 \right\}^{m}$ it outputs $n_{0^m}, n_{0^{m-1}1}, \ldots, n_{1^m}$  where for each string $y \in \left\{ 0,1 \right\}^{m}$, $n_y\in \left\{ 0,1 \right\}^{\lceil 1+\log(n) \rceil}$ represents $\abs{ \left\{ j \in [n] \mid x_j = y \right\} }$ in binary.
	The size of the circuit \cktFastCount{n}{m} is $\mathcal{O}(nm^2)$ and depth $\mathcal{O}(\log(n) + m \log(m))$.
	\label{lemma:fast_count}
\end{lemma}

We also provide a family of boolean circuits which sort the input integers by their first $k$ bits only. One can view this as sorting (key, value) pairs, where keys have $k$ bits and values have $m-k$ bits.
For the special case of $k=1$ (that is partially sorting the numbers by a single bit) the problem is equivalent to routing in \emph{super-concentrators} (see Section~\ref{sec:technique}), and we use super-concentrators of Pippenger~\cite{pippenger1996self} as our building block.
We get size improvement over the result of Asharov~et~al.~\cite{asharov2021sorting} while achieving also poly-logarithmic depth.

\begin{theorem}
	For any integers $n,m,k\ge 1$ where $k \leq m$ and $k \leq \log(n) / 11$ there is a circuit
	$$\cktSort{n}{m}{k} \fromto{nm}{nm}$$
	which partially sorts $n$ numbers each of $m$ bits by their first $k$ bits.
	The circuit \cktSort{n}{m}{k} has size $\mathcal{O}(k n m (1 + \log^*(n) - \log^*(m)))$ and depth $\mathcal{O}(\log^3(n))$.
	\label{thm:sort_by_k_bits}
\end{theorem}

\subsection{Our Techniques}
\label{sec:technique}

One can take AKS sorting networks and turn them into circuits of size  $\mathcal{O}(n m \log(n))$ and depth $\mathcal{O}(\log(m) \log(n))$.
For $m=o(\log(n))$ this is sub-optimal as shown by Asharov~et~al.~\cite{asharov2021sorting}. Asharov~et~al. show how to reduce the problem of sorting $m$-bit integers according to the first $k$ bits into the problem of sorting $m$-bit integers according to just single bit. Sorting according to single bit is essentially equivalent to routing in super-concentrators.

Super-concentrators have been studied originally by Valiant with the aim of proving circuit lower bounds.
A super-concentrator is a graph with two disjoint subsets of vertices $A, B \subseteq V(G)$, called inputs and outputs, with the property that for any set $S \subseteq A$ and $T \subseteq B$ of the same size there is a set of vertex disjoint paths from each vertex of $S$ to some vertex of $T$.
Pippenger~\cite{pippenger1996self} constructs super-concentrators with a linear number of edges and an algorithm that on input describing $S$ and $T$ outputs the list of edges forming the disjoint paths between $S$ and $T$. This can be turned into a circuit of size $\mathcal{O}(n \log(n))$ and depth~$\mathcal{O}(\log^2(n))$.

The result of Pippenger~\cite{pippenger1996self} can be used to build a circuit sorting by one bit, but the circuit will be larger than we want (see Corollary~\ref{col:pippenger_sort}.)
Thus, Asharov~et~al.~\cite{asharov2021sorting} used the technique of Pippenger rather than his result to design a circuit sorting by one bit, and iterate it to sort by $k$ bits.
Our technique differs substantially from that of Asharov~et~al. yet, we use the circuits from AKS networks and from Pippenger's super-concentrators as black box. 

To sort $m$-bit integers for $2^m \ll n$ our approach is to count the number of occurrences of each number in the input. This compresses the input from $n m$ bits into $2^m \log(n)$ bits. We can then decompress the vector back to get the desired output. So the main challenge is to construct counting (compressing) circuits of size $\mathcal{O}(nm^2)$. Interestingly, we use the sorting circuits derived from AKS networks to do that. But to avoid the size blow-up we don't use them on all of the integers at once but on blocks of integers of size $2^{8m}$. Then the $\mathcal{O}(\log(n))$ overhead of the circuits turns into the acceptable $\mathcal{O}(m)$ overhead. Each sorted block is then subdivided into parts of size $2^{2m}$. Clearly, most parts in each block will be monochromatic, they will contain copies of the same integer. There will be at most $2^m$ non-monochromatic parts. We \emph{move} the parts 
within a block to one side using another application of the AKS sorting circuit. Then we can afford to build a fairly expensive counting circuit for the small fraction of non-monochromatic parts, while cheaply counting the monochromatic parts. Summing the results by linear size circuit gives us the desired compression. Our decompression essentially mirrors the compression.

We also design a circuit to sort according to a single bit improving the parameters of Asharov~et~al.~\cite{asharov2021sorting}. We take the circuit of Pippenger as basis and apply it iteratively to larger and larger blocks of inputs. Again we start from blocks of size $2^{O(m)}$, and increase the size of the blocks exponentially at each iteration. We use Pippenger's circuit to sort each block by the bit. When we split the block into parts, only one will be monochromatic. Merging multiple blocks into one gives a mega-block with only a small fraction of non-monochromatic parts. These non-monochromatic parts can be separated from monochromatic ones, re-sorted, and re-partitioned to give only one non-monochromatic part in the mega-block. Each part takes on the role of an ``$m$''-bit integer in the next iteration. Iterating this process leads to the desired result.

To sort according to the first $k$ bits we use the one-bit sorting similarly to Asharov~et~al.~\cite{asharov2021sorting}. Thanks to our efficient sorting circuits for $m$-bit integers to sort the $k$-bit keys, we can avoid the use of median finding circuits.

\paragraph*{Organization.}
In the next section we review our notation.
We provide basic construction tools including na\"{\i}ve constructions of counting and decompression circuits in Section~\ref{sec:preliminaries}.
In Section~\ref{sec:sorting} we recall basic facts on AKS sorting networks and related sorting circuits.
In Section~\ref{sec:sorting_strings} we prove our main result by constructing efficient counting and decompression circuits.
Finally, we provide a construction of partial sorting circuits for Theorem~\ref{thm:sort_by_k_bits} in Section~\ref{sec:sorting_with_payloads}.

\section{Notation}
\label{sec:notation}

In this paper $\N$ denotes the set of natural numbers, and for $1\le a\le b \in \N$, $[a,b]=\{a,a+1,\dots,b\}$ and $[a]=\{1,\dots,a\}$.
All logarithms are base two unless stated otherwise.
For $m \in \N$, $\left\{ 0,1 \right\}^{m}$ is the set of all binary strings of length $m$.
A string $x \in \left\{ 0,1 \right\}^{m}$, $x=x_1x_2\cdots x_m$, represents the number $\sum_{j \in [m]} x_j 2^{m-j}$ in binary, and we often identify
the string with that number. (As the same integer has multiple binary representations differing in the number of leading zeroes, the number of leading zeroes should be clear from the context.) 
The most significant bit of $x=x_1x_2\cdots x_m$ is $x_1$ and the least significant bit of $x$ is $x_m$.
Symbol $\circ$ denotes the concatenation of two strings.
For strings $x, y \in \left\{ 0,1 \right\}^{m}$, $x \oplus y$ denotes the bit-wise XOR of $x$ and $y$, $x \wedge y$ denotes the bit-wise AND, and $x \vee y$ the bit-wise OR.

We assume the reader is familiar with boolean circuits (see for instance the book of Jukna~\cite{jukna2012boolean}). We assume boolean circuits consist of gates computing binary AND and OR, and unary gates computing negation. For us, boolean circuits might have multiple outputs so a circuit with $n$ inputs and $m$ outputs computes a function $f\fromto{n}{m}$.
We usually index a circuit family by multiple integral parameters. 
Inputs and outputs of boolean circuits are often interpreted as sequences of substrings, e.g., a circuit $C_{n, m} \fromto{nm}{nm}$ 
is viewed as taking $n$ binary strings of length $m$ as its input, and similarly for its output.
We say a circuit family $(C_n)_{n\in\N}$ is uniform, if there is an algorithm that on input $1^n$ outputs the description of the circuit $C_n$ in time polynomial in $n$.

\section{Preliminaries}
\label{sec:preliminaries}

Here we review some of the circuits for basic primitives that we will use in our later constructions. Most of them are well known facts but for the others we provide proofs for the sake of completeness.

\newcommand{\cktPlusSizeDepth}[1]{\cktPlus{{#1}} has size $\mathcal{O}\left( {#1} \right)$ and depth $\mathcal{O}\left( \log\left({#1}\right) \right)$}
\begin{lemma}[Addition]
	There is a uniform family of boolean circuits $\cktPlus{m} \fromto{2m}{m+1}$
	that given $x, y \in \left\{ 0,1 \right\}^{m}$ representing two numbers in binary outputs their sum $x+y \in \left\{ 0,1 \right\}^{m+1}$.
	The circuit \cktPlus{m} has size $\Theta\left( m \right)$ and depth $\Theta( \log(m) )$.
	\label{lem:sum_two_numbers}
\end{lemma}

\newcommand{\cktDifferenceSizeDepth}[1]{\cktDifference{ {#1} } has size $\Theta\left( {#1} \right)$ and depth $\Theta\left( \log\left( {#1} \right) \right)$}
\begin{lemma}[Subtraction]
	There is a uniform family of boolean circuits $\cktDifference{m} \fromto{2m}{m}$
	that given $x, y \in \left\{ 0,1 \right\}^{m}$ representing two numbers in binary outputs the absolute value of their difference $|x-y| \in \left\{ 0,1 \right\}^{m}$.
	The circuit \cktDifferenceSizeDepth{m}.
	\label{lem:dif_two_numbers}
\end{lemma}

\newcommand{\cktSumSizeDepthParenthesis}[2]{\cktSum{ {#1} }{ {#2} } has size $\Theta\left( \left( {#1} \right) \left( {#2} \right) \right)$ and depth $\Theta\left( \log\left( {#1} \right) + \log\left( {#2} \right) \right)$}
\newcommand{\cktSumSizeDepthTODO}[2]{\cktSum{ {#1} }{ {#2} } \todo[inline,color=green!40]{\cktSum{ {#1} }{ {#2} } size $\Theta\left( \left( {#1} \right) \left( {#2} \right) \right)$, depth $\Theta\left( \log\left( {#1} \right) + \log\left( {#2} \right) \right)$}}
\begin{lemma}[Summation]
	There is a uniform family of boolean circuits
	$$\cktSum{n}{m} \fromto{n \cdot m}{\lceil \log(n) \rceil + m}$$
	that given $x_1, x_2, \ldots, x_n \in \left\{ 0,1 \right\}^{m}$ interpreted as $n$ numbers, each of $m$ bits, outputs their sum $\sum_{j = 1}^{n} x_j$.
	The circuit \cktSum{n}{m} has size $\Theta(nm)$ and depth $\Theta\left( \log(n) + \log(m) \right)$.
	\label{lem:sum_n_numbers}
\end{lemma}

\begin{proof}
	We sketch the construction following the technique of Wallace~\cite{wallace1964suggestion}.
	Given three numbers $x, y, z \in \left\{ 0,1 \right\}^k$ in constant depth and using $\Theta(k)$ gates we can compute $p, q \in \left\{ 0,1 \right\}^{k+1}$ such that $x + y + z = p + q$.
	Here, $p$ is the coordinate-wise addition without carry, i.e., $0 \circ (x \oplus y \oplus z)$, and $q$ is the carry, i.e., $((x \wedge y) \vee (x \wedge z) \vee (y \wedge z)) \circ 0$.
	Thus as long as there are at least three numbers to sum we can use this to transform $x, y, z$ which takes $3k$ bits into $p, q$ which take $2k + 2$ bits and continue summing those.
	Doing this in parallel for disjoint triples of summants after $\mathcal{O}(\log_{3/2}(n)) = \mathcal{O}(\log(n))$ rounds we are left with just two numbers and we sum those using Lemma~\ref{lem:sum_two_numbers}.
\end{proof}

\begin{lemma}[Comparator]
	There is a uniform family of boolean circuits $\cktSwitch{m} \fromto{2m}{2m}$
	that given two numbers $x, y \in \left\{ 0, 1 \right\}^{m}$ outputs these two numbers sorted as integers, i.e., $\min(x, y) \circ \max(x, y)$.
	The size of the circuit \cktSwitch{m} is $\Theta(m)$ and depth is $\Theta(\log(m))$.
	\label{lem:switch}
\end{lemma}

Technique similar to the proof of the next lemma will be used also later in the proofs of Lemma~\ref{lemma:fast_count} and Lemma~\ref{lemma:fast_decompress} in order to achieve smaller circuit size. The main idea is to split inputs into smaller blocks and process the blocks independently by smaller circuits.

\newcommand{\cktOnesSizeDepthTODO}[1]{\cktOnes{ {#1} } \todo[inline,color=green!40]{\cktOnes{ {#1} } size $\Theta\left( 2^{\left( {#1} \right)} \right)$, depth $\Theta\left( {#1} \right)$}}
\begin{lemma}[Binary to unary]
	There is a uniform family of boolean circuits
	$$\cktOnes{b} \fromto{b+1}{2^b}$$
	such that for any number $x \in \left\{ 0,1 \right\}^{b+1}$ represented in binary the output consists of 
        $x$ ones followed by $2^{b} - x$ zeroes, provided $x \leq 2^b$.
        The circuit \cktOnes{b} has size $\Theta(2^{b})$ and depth $\Theta(\log(b))$.
	\label{lem:ones}
\end{lemma}

\begin{proof}
	We first show how to construct a uniform family of boolean circuits $\left( \cktOnesDeep{b} \right)$ which computes the same function, has the same size but depth $\mathcal{O}(b)$. Then we use \cktOnesDeep{\log(b)} to construct the desired circuit \cktOnes{b}.

	The main idea of the construction of \cktOnesDeep{b} is to recursively split the number $x$ into two numbers $x_L, x_R$ which describe how many bits set to one there should be in the first and the second half of the output.

	Each of the two numbers $x_L, x_R$ will be represented by $b$ bits with the convention that if the most significant bit is equal to one then the number is a power of two (corresponding to all output bits in this part of the output set to one).
	We recursively split the numbers $x_L, x_R$ in the same fashion until the numbers are represented by a single bit each at which point they will represent the output bits.
	We set
	\begin{align*}
		x_L &= \min(2^{b-1}, x) \\
		x_R &= \min(2^{b-1}, \max(0, x - 2^{b-1}))
	\end{align*}
	note that if the number $x$ is represented by $b+1$ bits ($x \in \left\{ 0,1 \right\}^{b+1}$) then the numbers $x_L, x_R$ can be represented by $b$ bits ($x_L, x_R \in \left\{ 0,1 \right\}^{b}$) as both of them represent at most half of $x$.
	Given $x \in \left\{ 0,1 \right\}^{b+1}$ we can compute the maximum and minimum defining $x_L, x_R$ by inspecting the two most significant bits of $x$:
	\begin{itemize}

		\item  If the most significant bit of $x$ is set to one (thus $x \ge 2^b$) we set $x_L = x_R = x/2$ each a power of two with the most significant bit set to one (and represented by a binary string~$10^{b-1}$).

		\item  If the most significant bit of $x$ is set to zero and the second most significant bit is set to one, then $x_L$ will be set to the binary number $10^{b-1}$ and $x_R$ will be $x - x_L$ (a copy of $x$ without the second most significant bit of~$x$).

		\item  If the two most significant bits of $x$ are equal to zero then $x_L = x$ (represented by one less bit than $x$) and $x_R = 0$.

	\end{itemize}
	See Figure~\ref{fig:ones_example} for an example of splitting of $x$ into $x_L, x_R$.

	\begin{figure}[h]
		\centering
		\begin{tikzpicture}
			\node [rectangle,draw]{0101} [level distance=10mm,sibling distance=50mm] child {
				node [rectangle,draw]{100} [level distance=10mm ,sibling distance=25mm] child {
					node [rectangle,draw] {10} [level distance=10mm ,sibling distance=15mm] child {
						node [rectangle,draw] {1}
					} child {
						node [rectangle,draw]{1}
					}
				} child {
					node [rectangle,draw] {10} [level distance=10mm ,sibling distance=15mm] child {
						node [rectangle,draw] {1}
					} child {
						node [rectangle,draw]{1}
					}
				}
			} child {
				node [rectangle,draw]{001} [level distance=10mm ,sibling distance=25mm] child {
					node [rectangle,draw] {01} [level distance=10mm ,sibling distance=15mm] child {
						node [rectangle,draw] {1}
					} child {
						node [rectangle,draw]{0}
					}
				} child {
					node [rectangle,draw] {00} [level distance=10mm ,sibling distance=15mm] child {
						node [rectangle,draw] {0}
					} child {
						node [rectangle,draw]{0}
					}
				}
			};
		\end{tikzpicture}
		\caption{
			An example of splitting numbers where $b = 3$.
			The input number $x = 5$ is represented as $0101$ and is split into $x_L = 100, x_R = 001$ which are themselves split recursively.
			The bottom nodes form the output.
		}
		\label{fig:ones_example}
	\end{figure}
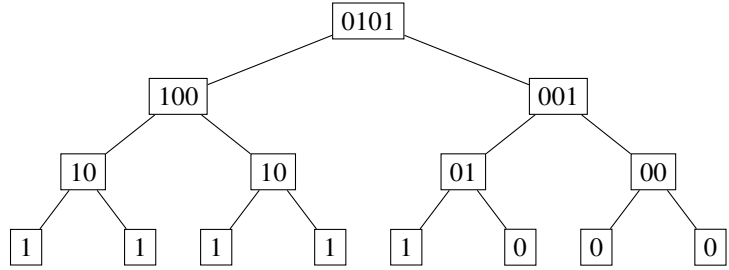

	Thus we can compute the transformation $x \mapsto (x_L, x_R)$ where $x \in \left\{ 0,1 \right\}^{b+1}$ and $x_L, x_R \in \left\{ 0,1 \right\}^{b}$ using a circuit of size $\Theta(b)$ and depth $\Theta(1)$.
	Then each of the numbers $x_L, x_R$ is again split into two, etc. until we get single bit numbers which represent the final output.
	The depth of the circuit \cktOnesDeep{b} is $\Theta(b)$ as each splitting can be done in constant depth. 
        If the circuit splitting $b+1$ bits into two $b$-bit numbers has size $s(b)\le cb+d$, for some universal constants $c$ and $d$, then
	the circuit \cktOnesDeep{b} has size:
	\begin{align*}
		s(b+1) + 2 s(b) + 4 s(b-2) + \ldots + 2^b s(1) &= \sum_{j = 0}^{b} 2^{j} s(b-j) \\
		&\leq  \sum_{j = 0}^{b} 2^{j} c(b-j) + 2^j d \\ 
		&\leq  c\left( 2^{b+2} - b - 1 \right) + 2^{b+1} d \\
		&= O(2^{b})
	\end{align*}

        To build the circuit \cktOnes{b} of depth $\mathcal{O}(\log(b))$ we proceed as follows.
	For any $y > 1$ we denote the largest power of two that is at most $y$ by $\ell(y) = \max\left\{ 2^{j} \mid j \in \N, 2^{j} \leq y \right\}$.
	We divide the output bits into blocks of $\ell(b)$ bits and for each block $j \in \left[ \frac{2^{b}}{\ell(b)} \right]$ of output bits with positions $[(j - 1) \ell(b) + 1, j \ell(b)]$ (counting positions from one) we compute if it should be constant (that is either constant zero when $x \leq (j-1) \ell(b)$ or constantly equal to one when $x > j \ell(b)$).
	This check for constant values can be done in each block by a circuit of size $\Theta(b)$ and depth $\Theta(\log(b))$.
	We compute \cktOnesDeep{\log(\ell(b))} with the input being the $\log(\ell(b))$ least significant bits of $x$. This circuit is of size $\mathcal{O}(b)$ and depth $\mathcal{O}(\log(b))$. In each block if the block should not be monochromatic then we use the output of that circuit as the output of the block, otherwise we use the appropriate constant one or zero copied $\ell(b)$-times as the output of the block.
\end{proof}

%

We will need a primitive that counts the number of occurrences of each string in the input.
A counting similar to Lemma~\ref{lem:count} appears in Appendix~A of the paper of Asharov~et~al.~\cite{asharov2021sorting}.
The construction of the counting circuit is rather straightforward, we just compare each input string $x_j$ with a given string $y$ getting an indicator bit set to one for equality and to zero for inequality and then sum the indicator bits.

\newcommand{\cktCountSizeDepthTODO}[2]{\cktCount{ {#1} }{ {#2} } \todo[inline,color=green!40]{\cktCount{ {#1} }{ {#2} } size $\Theta\left( \left( {#1} \right) \left( {#2} \right) 2^{\left( {#2} \right)} \right)$, depth $\Theta\left( \log\left( {#1} \right) +  \log\left( {#2} \right) \right)$}}
\begin{lemma}[Count]
	There is a uniform family of boolean circuits $\cktCount{n}{m} \fromto{n \cdot m}{2^m \lceil 1+\log(n) \rceil}$
	that given $x_1, x_2, \ldots, x_n \in \left\{ 0,1 \right\}^{m}$ counts the number of occurrences of each $y \in \left\{ 0,1 \right\}^{m}$ among the inputs, i.e., the circuit outputs $n_{0^m}, n_{0^{m-1}1}, \ldots, n_{1^m}$ where for each $y \in \left\{ 0,1 \right\}^{m}$, $n_y$ represents in binary $\abs{\left\{ j \in [n] \mid y = x_j \right\}}$ using $\lceil 1+\log(n) \rceil$ bits.
	The size of the circuit \cktCount{n}{m} is $\mathcal{O}(nm2^m)$ and depth $\mathcal{O}(\log(n) + \log(m))$.
	\label{lem:count}
\end{lemma}

\begin{proof}
	For each $y \in \left\{ 0,1 \right\}^{m}$ we build a sub-circuit computing the number of times $y$ occurs among the inputs $x_1,\dots,x_n$.
This is done by comparing $y$ to each $x_i$ in parallel, $i\in [n]$, to get an indicator bit whether they are equal. We obtain $n_y$ by summing up the 
indicator bits using the circuit \cktSum{n}{1} of size $\Theta(n)$ and depth $\Theta(\log(n))$ from Lemma~\ref{lem:sum_n_numbers}. Comparing $y$ to $x_i$ can be done by a circuit of size $\mathcal{O}(m)$ and depth $\mathcal{O}(\log(m))$. So we get $n_y$ using a circuit of size $\Theta(nm)$ and depth $\Theta(\log(n) + \log(m))$.
	Doing this for each $y \in \left\{ 0,1 \right\}^{m}$ in parallel we get a circuit of size $\Theta(nm 2^{m})$ and depth $\Theta(\log(n) + \log(m))$.
\end{proof}

We will need also an inverse operation for the counting.
To construct a circuit that decompresses the counts we would like to first compute the interval where a given string $x$ should appear and then get indicator bits for this interval. We can compute the interval using prefix sums of the counts.
To get the indicator bits for the interval we utilize the circuit from Lemma~\ref{lem:ones} which outputs a given number of bits set to one followed by bits set to zero.


\newcommand{\cktDecompressSizeDepthTODO}[2]{\cktDecompress{ {#1} }{ {#2} } \todo[inline,color=green!40]{\cktDecompress{ {#1} }{ {#2} } size $\Theta\left( \left( {#1} \right) \left( {#2} \right) 2^{ {#2} } + 2^{2({#2})} \log\left( {#1} \right) \right)$, depth $\Theta\left( {#2} + \log\left( {#1} \right) \right)$}}
\begin{lemma}[Decompress]
	There is a uniform family of boolean circuits $$\cktDecompress{n}{m} \fromto{\lceil 1+\log(n) \rceil 2^{m}}{n \cdot m}$$
	that \emph{decompresses its input} that is on input numbers $n_{0^m}, n_{0^{m-1}1}, \ldots, n_{1^m}$, each represented in binary by $\lceil 1+\log(n) \rceil$ bits, where $\sum_{x \in \left\{ 0,1 \right\}^{m}} n_x = s \leq n$, outputs the string
$$(\underbrace{00\cdots 0}_m)^{n_{00\cdots 0}} \circ  (\underbrace{00\cdots 0}_{m-1}1)^{ n_{00\cdots 01}} \circ  (\underbrace{00\cdots 0}_{m-2}10)^{ n_{00\cdots 010}} \circ  (\underbrace{00\cdots 0}_{m-2}11)^{ n_{00\cdots 011}} \circ \cdots \circ (\underbrace{11\cdots 1}_m)^{n_{11\cdots 1}} \circ (0^m)^{n-s}.$$ When $s>n$ the output might be arbitrary.
	The size of the circuit \cktDecompress{n}{m} is $\mathcal{O}(n m 2^m + 2^{2m} \log(n))$ and depth $\mathcal{O}(m + \log(\log(n)))$.
	\label{lem:decompress}
\end{lemma}

\begin{proof}
	Given $n_{0^m}, n_{0^{m-1}1}, \ldots, n_{1^m}$ we can compute the total sum $s = \sum_{x \in \left\{ 0,1 \right\}^{m}} n_x$ and for each $y\in \{0,1\}^m$, the number $p_y$ of binary strings before the first occurrence of $y$, i.e., $p_y = \sum_{x \in \left\{ 0,1 \right\}^m \colon x < y} n_x$.
	Each of the numbers $p_y$ can be computed using the circuit \cktSum{y}{\lceil 1+\log(n) \rceil} from Lemma~\ref{lem:sum_n_numbers} of size $\mathcal{O}(2^{m} \log(n))$ and depth $\mathcal{O}(m + \log \log(n))$. Similarly for $s$.
	Thus we can get all numbers $p_y$ in parallel by a circuit of size $\mathcal{O}(2^{2m} \log(n))$.
	A given string $y \in \left\{ 0,1 \right\}^{m}$, $y\neq 1^m$, should appear at each position $j \in [p_y + 1, p_{y+1}]$. Let $I_{y} \in \{0,1\}^n$ be the indicator vector of positions where $y$ should appear in the output.
	We can use $\cktOnes{\lceil 1+\log(n) \rceil}(p_y) \oplus \cktOnes{\lceil 1+\log(n) \rceil}(p_{y+1})$ to calculate $I_y$ for each $y\neq 1^m$. For $y = 1^{m}$, $I_y = \cktOnes{\lceil 1+\log(n) \rceil}(p_y) \oplus \cktOnes{\lceil 1+\log(n) \rceil}({s})$.
	The size of \cktOnes{\lceil 1+\log(n) \rceil} is $\Theta(n)$. As there are $2^m$ different $y$'s, we need a circuit of size $\Theta(n2^m)$ and depth $\Theta(\log \log(n))$ to calculate all $I_y$'s.

        If $x_1,x_2,\dots,x_n$ are the output integers, for each output position $j \in [n]$, we calculate the $k$-bit of $x_j$ as
	\begin{align*}
		\bigvee_{y \in \left\{ 0,1 \right\}^{m}} ((I_y)_j \wedge y_k)
	\end{align*}
	To compute all these ORs we need a circuit of total size $\Theta(nm2^{m})$ and depth $\Theta(m)$.
\end{proof}

\section{Sorting Circuits from AKS Sorting Networks}
\label{sec:sorting}

In this section we recall the construction of circuits for sorting from the Ajtai-Koml\'os-Szemer\'edi sorting networks. They will serve as the basic primitive for our later constructions. 

\paragraph*{Sorting networks.}
Sorting networks model parallel algorithms that sort values using only comparisons.
A sorting network consists of $n$ wires and $s$ comparators. The wires extend from left to right in parallel.
Each wire carries an integer from left to right. Any two wires can be connected by a comparator at any point along their length. The comparator
swaps the values carried along the two wires if the higher wire carries a higher value at that point otherwise it has no effect. The sorting network
should be such when we input arbitrary integers to the wires on the left, the integers always exit in sorted order from top to bottom.
The \emph{depth} of a sorting network is the maximum number of comparators a value can encounter on its way.
A figure of a small sorting network is given in Figure~\ref{fig:aks_example}.
       \begin{figure}[t]
                \centering
               \begin{tikzpicture}[
                                dot/.style = {circle, fill, minimum size=#1, inner sep=0pt, outer sep=0pt},
                                dot/.default = 5pt  
                    ]
                        \node (x) at (1,1) {$x$};
												\node (y) at (1,0.5) {$y$};
                        \node (z) at (1,0) {$z$};
												\node[dot] (minxy) at (2,1) {}; \node[dot] (maxxy) at (2,0.5){}; \draw (minxy) -- (maxxy);
                        \draw (x) -- (minxy);
                        \draw (y) -- (maxxy);
                        \node (minxylabel) at (3.5,1) {$\min(x,y)$};
												\node (maxxylabel) at (3.5,0.5) {$\max(x,y)$};
                        \draw (minxy) -- (minxylabel);
                        \draw (maxxy) -- (maxxylabel);
                        \node[dot] (minxyz) at (5,1) {}; \draw (minxylabel) -- (minxyz);
                        \node[dot] (maxxz) at (5,0) {};
                        \draw (z) -- (maxxz);
                        \draw (minxyz) -- (maxxz);
                        \node (maxxzlabel) at (6.5,0) {$\max(\min(x,y),z)$};
												\node[dot] (midxyz) at (8,0.5) {};
                        \node[dot] (maxxyz) at (8,0) {};
                        \draw (midxyz) -- (maxxyz);
                        \draw (maxxylabel) -- (midxyz);
                        \draw (maxxz) -- (maxxzlabel);
                        \draw (maxxzlabel) -- (maxxyz);
                        \node (minlabel) at (9.5,1) {$\min(x,y,z)$};
												\node (medlabel) at (9.5,0.5) {$\text{median}(x,y,z)$};
                        \node (maxlabel) at (9.5,0) {$\max(x,y,z)$};
                        \draw (minxyz) -- (minlabel);
                        \draw (midxyz) -- (medlabel);
                        \draw (maxxyz) -- (maxlabel);
                \end{tikzpicture}
                \caption{
                        An example of a sorting network with three inputs (the horizontal lines), three comparators (the vertical lines), and depth three.
                        The inputs on the left are numbers $x,y,z$ and after each comparator we noted what is on the horizontal line.
                        Note that the bottom most output is $\max(\max(x,y),\max(\min(x,y),z)) = \max(x,y,z)$ and the middle one is $\min(\max(x,y),\max(\min(x,y),z))$ which is the median.
                }
                \label{fig:aks_example}
        \end{figure}
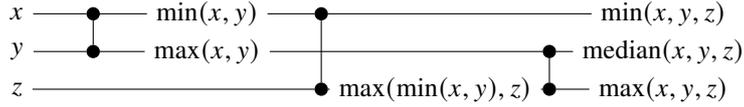
For a formal definition see, e.g., \cite{aks_1983sorting}. 
Observe that if the depth of a sorting network is $d$ and the number of inputs is $n$ then there are at most $s \le nd$ comparators.
Ajtai, Koml\'os and Szemer\'edi~\cite{aks_1983sorting} established the existence of sorting networks of logarithmic depth.

\begin{theorem}[AKS~\cite{aks_1983sorting}]
	For any integer $n\ge 1$, there is a sorting network for $n$ integers of depth $\mathcal{O}(\log(n))$.
	\label{thm:aks_network}
\end{theorem}

\paragraph*{Sorting circuits.} Here we give a precise definition of sorting by a circuit. First we consider a circuit sorting $n$ integers, each of them $m$ bits long.

\begin{definition}[Sort]
	Let $n, m \in \N$, and $\left( C_{n, m} \right)$ be a family of boolean circuits. 
	We say that the \emph{circuit $C_{n, m} \fromto{nm}{nm}$ sorts its input} interpreted as $n$ integers $x_1, x_2, \ldots, x_n$ each represented by $m$ bits if it outputs $y_1, y_2, \ldots, y_n \in \left\{ 0,1 \right\}^{m}$ such that:
			\begin{enumerate}

				\item  The outputs are sorted: For any $i < j \in [n]$, $y_i \leq y_j$.

				\item  The inputs and outputs form the same multiset: For each $j \in [n]$, $\abs{ \left\{ i \in [n] \mid y_i = x_j \right\} } = \abs{ \left\{ i \in [n] \mid x_i = x_j \right\} }$.

			\end{enumerate}
	\label{def:sorting_circuit}
\end{definition}

An immediate consequence of the existence of AKS sorting networks is the existence of shallow sorting circuits, since
by Lemma~\ref{lem:switch}, each comparator can be replaced by a small circuit:

\newcommand{\cktAKSSizeDepthTODO}[2]{\cktAKS{ {#1} }{ {#2} } \todo[inline,color=green!40]{\cktAKS{ {#1} }{ {#2} } size $\Theta\left( \left( {#1} \right) \left( {#2} \right) \log\left( {#1} \right) \right)$, depth $\Theta\left( \log\left( {#1} \right) \log\left( {#2} \right) \right)$}}
\begin{corollary}
	There is a family of boolean circuits $\cktAKS{n}{m} \fromto{n \cdot m}{n \cdot m}$
	that on an input $x_1, x_2, \ldots, x_n \in \left\{ 0,1 \right\}^{m}$ sorts these numbers.
	The size of the circuit \cktAKS{n}{m} is $\mathcal{O}(n m \log(n))$ and depth $\mathcal{O}(\log(n) \log(m) )$.
	\label{col:aks_circuit}
\end{corollary}

We also need circuits that sort the $n$ input integers, each of $m$ bits, by the $k$ most significant bits where $k<m$. Such sorting can be thought of as sorting (key, value) pairs, where keys are $k$-bit long and values $(m-k)$-bit long. Formally it can be defined as follows:

\begin{definition}[Partial Sort]
	Let $n, m, k \in \N$, be such that $k<m$, and let $\left( C_{n, m, k} \right)$ be a family of boolean circuits. 
	We say that the \emph{circuit $C_{n, m, k} \fromto{nm}{nm}$ partially sorts by the first $k$ bits its input} interpreted as $n$ integers $x_1, x_2, \ldots, x_n$ each represented by $m$ bits if it outputs $y_1, y_2, \ldots, y_n \in \left\{ 0,1 \right\}^{m}$ such that:
			\begin{enumerate}

				\item  The outputs are partially sorted: For any $i < j \in [n]$, $(y_i)_1 (y_i)_2 \cdots (y_i)_k \leq (y_j)_1 (y_j)_2 \cdots (y_j)_k$.

				\item  The inputs and outputs form the same multiset: For each $j \in [n]$, $\abs{ \left\{ i \in [n] \mid y_i = x_j \right\} } = \abs{ \left\{ i \in [n] \mid x_i = x_j \right\} }$.

			\end{enumerate}
	\label{def:partial_sorting_circuit}
\end{definition}

Using a circuit of size $\mathcal{O}(m)$ and depth $\mathcal{O}(\log(k))$ implementing a comparator which swaps two $m$-bit integers based only on the first $k$ bits we get the following variant of the previous corollary.

\newcommand{\cktAKSPartialSizeDepthTODO}[3]{\cktAKSPartial{ {#1} }{ {#2} }{ {#3} } \todo[inline,color=green!40]{\cktAKSPartial{ {#1} }{ {#2} }{ {#3} } size $\Theta\left( \left( {#1} \right) \left( {#2} \right) \log\left( {#1} \right) \right)$, depth $\Theta\left( \log\left( {#1} \right) \log\left( {#3} \right) \right)$}}
\begin{corollary}
	There is a family of boolean circuits $\cktAKSPartial{n}{m}{k} \fromto{n \cdot m}{n \cdot m}$, for $k \leq m$ and $k \leq \log(n)$,
	that on input $x_1, x_2, \ldots, x_n \in \left\{ 0,1 \right\}^{m}$ partially sorts these numbers according to their $k$ most significant bits.
	That is if $y_i, y_j$ are two output numbers where $i < j$ then we have $\lfloor y_i / 2^{m-k} \rfloor \leq \lfloor y_j / 2^{m-k} \rfloor$.
	The size of the circuit \cktAKSPartial{n}{m}{k} is $\mathcal{O}(n m \log(n))$ and depth $\mathcal{O}(\log(n) \log(k))$.
	\label{col:aks_partial_circuit}
\end{corollary}

\section{Sorting $n$ Binary Strings of Length $m$}
\label{sec:sorting_strings}

Here we present a sorting circuit for short numbers.
The construction consists of two circuits.
The first circuit counts the number of occurrences of various strings (as stated in Lemma~\ref{lemma:fast_count}) and the second circuit 
decompresses these counts.
Both of these constructions use heavily the following technique: we divide the problem into blocks which can be efficiently sorted using the AKS-based circuit. These blocks will be of size between $2^{O(m)}$ and $n/2^{O(m)}$ where $m$ is the binary length of the input integers.

Thus when we sort the numbers inside each block and subdivide the block into parts, then by the pigeon-hole principle, most of the parts will be monochromatic (containing copies of a single string only).
We can then separately count the strings in monochromatic parts (count the first string and then multiply that by the length of the part) and in the non-monochromatic parts (there are not that many strings in total in non-monochromatic parts).
However a priori we do not know which parts will be monochromatic and which will be not. To save on circuitry we use sorting (on whole parts) to move the non-monochromatic parts aside. We build the (expensive) counting circuits only for non-monochromatic parts. 

\begin{proof}[Proof of Lemma~\ref{lemma:fast_count}]
	For the sake of simplicity let us assume that $n$ is a power of two so, it is divisible by $2^{8m}$. (By our assumption $n \geq 2^{10m}$, thus if $n$ is not a power of two take the circuit for the closest power of two larger than~$n$ and feed ones for the extra input bits.)
	We partition the input into $n / 2^{8m}$ blocks each consisting of $2^{8m}$ numbers.
	We sort each block by the circuit \cktAKS{2^{8m}}{m} of size $\mathcal{O}(2^{8m} m \log(2^{8m})) = \mathcal{O}(2^{8m} m^2)$ and depth $\mathcal{O}(m\log(m))$ as given in Corollary~\ref{col:aks_circuit} .
	Thus for this phase we need a circuit of total size $\mathcal{O}(n m^2)$.

	Then we subdivide each block into $2^{6m}$ parts each consisting of $2^{2m}$ numbers.
	Observe that most of these parts are monochromatic: a part is monochromatic if it contains $2^{2m}$ copies of a single $m$-bit number.
	We can upper bound the number of non-monochromatic parts by $2^{m}$.
	We can add a single indicator bit to each part indicating whether this part is monochromatic. As the parts are sorted it is enough to compare the first and last number in each part and set the bit to $1$ if the numbers are equal and to $0$ otherwise. 
We sort the parts prefixed by their indicator bit using the circuit \cktAKSPartial{2^{6m}}{1 + m 2^{2m}}{1} from Corollary~\ref{col:aks_partial_circuit} to move all non-monochromatic parts to the front of each block.
	Thus the total size of the circuit sorting parts inside each block is $\mathcal{O}\left( \frac{n}{2^{8m}} (2^{6m}) (1 + m 2^{2m}) 6m \right) = \mathcal{O}(nm^2)$ and depth $\mathcal{O}(m)$.
	We call the first $2^{m}$ parts of each block \emph{potentially non-monochromatic}. The other parts are \emph{definitely monochromatic}.

	From each definitely monochromatic part we take the first $m$-bit number and we count them.
	This can be done by the circuit \cktCount{\frac{n}{2^{8m}} (2^{6m} - 2^{m})}{m} from Lemma~\ref{lem:count} of size $\mathcal{O}\left( \left( \frac{n}{2^{2m}} - \frac{n}{2^{7m}} \right) m 2^{m} \right) \leq \mathcal{O}(nm)$ and depth $\mathcal{O}(\log(n) + \log(m))$.
	By multiplying each count by $2^{2m}$ (that is by appending $2m$ zeroes) we get the number of occurrences of each number in the definitely monochromatic parts.

	As there are relatively few (exactly $\frac{n}{2^{8m}} 2^{m} 2^{2m}$) numbers overall in potentially non-monochromatic parts we can use the circuit \cktCount{n / 2^{5m}}{m} from Lemma~\ref{lem:count} to count those numbers by a circuit of size $\mathcal{O}\left( \frac{n}{2^{5m}} m 2^{m} \right) \leq \mathcal{O}(nm)$ and depth $\mathcal{O}(\log(n) + \log(m))$.

        Thus we get two vectors of counts for numbers in potentially non-monochromatic and definitely monochromatic blocks.
	Finally, we add the two vectors of $2^{m}$ numbers each consisting of at most $\lceil 1+\log(n) \rceil$ bits to get the resulting counts.
        This uses a circuit of size $\mathcal{O}(m2^m)=O(n)$ and depth $\mathcal{O}(\log \log(n))$. Thus, the overall size of the circuit is $\mathcal{O}(nm^2)$ and depth $\mathcal{O}(\log(n) + m\log(m))$.
\end{proof}

\begin{lemma}
	For integers $n, m\ge 1$ such that $m \leq \log(n) / 11$, 
        there is a family of boolean circuits $$\cktFastDecompress{n}{m} \fromto{\lceil 1+\log(n) \rceil 2^{m}}{n \cdot m}$$
	that decompresses its input as in Lemma~\ref{lem:decompress}.
	The size of \cktFastDecompress{n}{m} is $\mathcal{O}(nm^2)$ and its depth is $\mathcal{O}(m \log(m) + \log \log(n))$.
	\label{lemma:fast_decompress}
\end{lemma}

The construction of the decompression circuit mirrors the counting circuit albeit it is somewhat simpler with a different choice of parameters.
We separately decompress monochromatic blocks (by decompressing just a single string from each block and then creating the right number of copies) and the strings from non-monochromatic blocks (as there are not many of those).
We then use partial sorting to rearrange the blocks in the proper order to construct a sorted sequence.

\begin{proof}
	For the sake of simplicity let us assume that $n$ is a power of two and let us set $k = n / 2^{8m}$. (Thus $k$ is an integer.) We will think  of the output as partitioned into $2^{8m}$ blocks of size $k$.
	As in the proof of Lemma~\ref{lem:decompress} we compute the prefix sums
	\begin{align*}
		p_{x} &= \sum_{y \in \left\{ 0,1 \right\}^{m} \colon y < x} n_{y} & \text{for each $x \in \left\{ 0,1 \right\}^{m}$}
	\end{align*}
	and we set $p_{2^{m}} = n$. (Here, we identify $m$-bit strings $x$ and $y$ with integers they represent.)
	We can compute each $p_x$ using the circuit \cktSum{2^{m}}{1+\log(n)}, thus computing all of them using a circuit of size $\mathcal{O}(\log(n)2^{2m}) \leq \mathcal{O}(n)$ (by the assumption $m \leq \log(n) / 11$) and depth $\mathcal{O}(m + \log \log(n))$.
	Thus the string $x \in \left\{ 0,1 \right\}^{m}$ should appear at output positions $[p_{x} + 1, p_{x + 1}]$.
	For any $x \in \left\{ 0,1 \right\}^{m}$ we set:
	\begin{align*}
		r_{x} &= \left( \left( k - \left( p_x \mod k \right) \right) \mod k \right) + \left( p_{x+1} \mod k \right) \\
		q_{x} &= \frac{n_x - r_x}{k}
	\end{align*}
	The meaning is that if we partition the output into blocks of $k$ consecutive numbers, then for any $x \in \left\{ 0,1 \right\}^{m}$ the number $r_x$ tells the number of times the string $x$ appears in non-monochromatic blocks. (These occurrences are located in at most two non-monochromatic blocks.)
	The number $q_{x}$ tells us in how many monochromatic blocks the string $x \in \left\{ 0,1 \right\}^{m}$ appears. Observe that $q_x$ is an integer. Since $n$ is a power of two, so is $k$, furthermore, $k$ is fixed for given $n$ and $m$, and thus computing mod $k$ and division by $k$ corresponds to selecting appropriate bits from the binary representation of numbers. All numbers $p_x$, $q_x$ and $r_x$ are integers represented by $1 + \log(n)$ bits.
Hence, each $q_x$ and $r_x$ can be computed from $n_x$ and $p_x$ by one circuit \cktPlus{1 + \log(n)} and two \cktDifference{1 + \log(n)}. The circuit computing values $q_x$ and $r_x$ for all $x$ has total size $\mathcal{O}(2^m \log(n))$ and depth $\mathcal{O}(\log \log(n))$.

	The following holds:
	\begin{align*}
		n_x &= k q_x + r_x \\
		\sum_{x \in \left\{ 0,1 \right\}^{m}} q_{x} &= \sum_{x \in \left\{ 0,1 \right\}^{m}}  \frac{n_x - r_x}{k} \leq n / k = 2^{8m} \\
		\sum_{x \in \left\{ 0,1 \right\}^{m}} r_{x} &\leq 2k 2^{m} = 2n / 2^{7m} \\
	\end{align*}


	We use circuit $\cktDecompress{2^{8m}}{m}(q_{0^m}, q_{0^{m-1}1}, \ldots, q_{1^{m}})$ from Lemma~\ref{lem:decompress} of size $\mathcal{O}\left( m 2^{9m} \right)$ and depth $\mathcal{O}\left( m \right)$ to decompress monochromatic blocks.
	We then just copy each resulting number $k$ times to create sorted monochromatic blocks.
	Last $2^{8m}-\sum_{x \in \left\{ 0,1 \right\}^{m}} q_{x}$ blocks contain zero padding corresponding to the numbers in non-monochromatic blocks. They will be merged with the non-monochromatic blocks obtained next.

	In order to properly match the non-monochromatic blocks to the padded zeroes we adjust the count $r_{0^m}$:
	\begin{align*}
		r'_{0^m} &= \left( 2n / 2^{7m} \right) - \sum_{x \in \left\{ 0,1 \right\}^m \colon x \neq 0^{m}} r_x
	\end{align*}
	using circuit \cktSum{2^{m}}{1+\log(n)} and \cktDifference{1 + \log(n)} of size $\mathcal{O}(n)$ and depth  $\mathcal{O}(m + \log \log(n))$.
	We use the circuit $\cktDecompress{2n / 2^{7m}}{m}(r'_{0^m}, r_{0^{m-1}1}, \ldots, r_{1^{m}})$ from Lemma~\ref{lem:decompress} to decompress the non-monochromatic blocks. The circuit is of size $\mathcal{O}\left( \left( 2n/2^{7m} \right) m 2^{m} + 2^{2m} \log\left( 2n / 2^{7m} \right) \right) \leq \mathcal{O}\left( nm / 2^{6m} \right)$ and of depth $\mathcal{O}(m + \log(\log(n)))$. (Here, we used our assumption $m \leq \log(n) / 11$, to bound  $n \geq 2^{11m}$ and $2^{2m} \leq n^{3/4} / 2^{6m}$.)

	Finally, we compute the bit-wise OR of the last $2^{m+1}$ blocks of the output from the previous step (monochromatic decompression) with the current output (non-monochromatic decompression). This way we get a sequence of $n$ numbers partitioned into blocks where each block corresponds to one of the blocks in the desired output. However, we still need to rearrange the blocks in the proper order. We will use partial sorting of the whole blocks to do that.

        For a given block let $x$ be the first number in that block. We prefix the block by a number $2x$ (represented by $m+1$ bits) if the block is monochromatic or the number $2x+1$ if the block is non-monochromatic. To determine whether the block is monochromatic we compare for equality the first and last number inside the block. We do this for each block. Thus each block of $k$ numbers is prefixed by an $m+1$ bit number.
Computing these prefixes requires a circuit of total size $\mathcal{O}(2^{8m} m) = O(n)$ and depth $\mathcal{O}(\log(m))$.
	We then use the \cktAKSPartial{2^{8m}}{(m+1) + km}{m+1} circuit of size $\mathcal{O}(nm^2)$ and depth $\mathcal{O}(m \log(m))$ to sort the blocks.	Finally, we ignore the $m+1$ bit prefixes of each block to get the desired output.
\end{proof}

\begin{proof}[Proof of Theorem~\ref{thm:main}]
	This is just a combination of Lemma~\ref{lemma:fast_count} with Lemma~\ref{lemma:fast_decompress}.
\end{proof}

Observe that the proofs of Lemma~\ref{lemma:fast_count} and Lemma~\ref{lemma:fast_decompress} do not depend on using specifically the AKS sorting.
In particular for the case of Lemma~\ref{lemma:fast_count} if there is a circuit that sorts input numbers that is linear in the number of input bits then there is a linear size circuit that counts these numbers.

\section{Partial Sorting by the First $k$ Bits in Poly-logarithmic Depth}
\label{sec:sorting_with_payloads}

Here we design a family of boolean circuits that partially sorts by the first $k$ bits out of $m$ bits which is asymptotically smaller than \cktAKSPartial{n}{m}{k}. We will need super-concentrators for our construction.

A directed acyclic graph $G = (V, E, A, B)$, where $V$ is the set of vertices, $E$ is the set of directed edges, and $A$ and $B$ are disjoint subsets
of vertices of the same size, is a \emph{super-concentrator} if the following hold:
The vertices in $A$ (\emph{inputs}) have in-degree zero, vertices in $B$ (\emph{outputs}) have out-degree zero, and for any $S \subseteq A$ and for any $T \subseteq B \colon |S| = |T|$ there is a set of pairwise vertex disjoint paths connecting each vertex from $S$ to some vertex in $T$.

We parametrize the super-concentrator by the number of input vertices $n$, and we measure its size by the number of edges. We want the graph to have as few edges as possible. The depth of the super-concentrator is the number of edges on the longest directed path.

Pippenger~\cite{pippenger1996self} shows a construction of super-concentrators of linear size and logarithmic depth. He constructs a family of super-concentrators $S_n$ for $n$ being the number of inputs, where the in-degree and out-degree of each vertex is bounded by some universal constant, the number of edges is linear in $n$, and the depth is $\mathcal{O}(\log(n))$.
Moreover there are finite automatons which for any $S \subset A, T \subset B \colon |S| = |T|$ when put on the vertices of the super-concentrator find the set of vertex disjoint paths from $S$ to $T$ in $\mathcal{O}(\log(n))$ iterations, each taking $\mathcal{O}(\log(n))$ steps, for the total number of $\mathcal{O}(n)$ steps of the automatons. We describe this construction using the language of circuits. The circuit on input of characteristic vector of $S$ and $T$ computes the set of $|T|$ vertex disjoint paths connecting $S$ and $T$. The circuit outputs the characteristic vector of the set of edges participating in the paths.

\newcommand{\cktRouteSizeDepthTODO}[1]{\todo[inline,color=green!40]{\cktRoute{ {#1} } size $\mathcal{O}\left( \left( {#1} \right) \log\left( {#1} \right) \right)$, depth $\mathcal{O}\left( \log\left( {#1} \right)^2 \right)$}}
\begin{theorem}[Pippenger~\cite{pippenger1996self}]
	There is a family of super-concentrators $S_n$ as described above and boolean circuits $\cktRoute{n} \fromto{2n}{|S_n|}$ of size $\mathcal{O}(n \log(n))$ and depth $\mathcal{O}(\log^2(n))$
that on input characteristic vector of any set $T \subseteq [n]$ and characteristic vector of any $S \subseteq [n]$ where $|T| = |S|$, outputs the characteristic vector of edges that form $|T|$ vertex disjoint paths between $S$ and $T$.
	\label{thm:pippenger_routing}
\end{theorem}

By routing $m$ bits along each path in the super-concentrator we can use the above circuit to build a circuit that partially sorts $m$-bit integers by their most significant bit.

\begin{corollary}
	There is a family of boolean circuits $\cktPSort{n}{m}{1} \fromto{n \cdot m}{n \cdot m}$
	that on input $x_1, x_2, \ldots, x_n \in \left\{ 0,1 \right\}^{m}$ partially sort these numbers according to their first most significant bit.
	The size of the circuit \cktPSort{n}{m}{1} is $\mathcal{O}(n m + n \log(n))$ and depth $\mathcal{O}(\log^2(n))$.
	\label{col:pippenger_sort}
\end{corollary}

\begin{proof}
We give a sketch of the proof. First, we will use the graph $S_n$ to get all inputs starting with one to the proper place. Then, using the same construction we will move all inputs starting by $0$ to the proper place. We transform the graph $S_n$ into a circuit by replacing each vertex of in-degree $d$ by a \emph{routing gadget} (circuit)
which takes $d$ $m$-bit inputs together with $d$ control bits, one bit for each of the $m$-bit inputs, and outputs the bit-wise OR of inputs for
which their control bit is set to 1. Such a routing gadget of size $\mathcal{O}(dm)$ and depth $\mathcal{O}(\log(d))$ can be easily constructed. If $(u,v)$ is the $j$-th
incoming edge of $v$ in $S_n$, we connect the $j$-th block of $m$ input bits of the routing gadget corresponding to $v$ to the output of the routing gadget of $u$. The routing gadgets of input vertices of $S_n$ are connected directly to the appropriate inputs of the sorting circuit.
The routing gadget will be used with at most single control bit set to one, thus it will route the corresponding input.

It remains to calculate paths that will route the integers starting with 1 in the above circuit in the desired way. For that, we calculate the sum $s$ of the most significant bits by which we are sorting using \cktSum{n}{1} from Lemma~\ref{lem:sum_n_numbers}, we expand it back using $\cktOnes{\lceil \log(n) \rceil + 1}(s)$, and reverse it to get the characteristic vector of a set $T$, where we want to route to. Together with the most significant bits of each input integer (which form the characteristic vector of $S$ from which we route) we feed this as an input to $\cktRoute{n}$. The output bits of $\cktRoute{n}$ are connected to the appropriate control bits of our routing gadgets. The sorted output will be obtained as the output of the $n$ routing gadgets corresponding to the output vertices of~$S_n$.

The size of the $\cktRoute{n}$ is $\mathcal{O}(n \log(n))$ and the total size of the circuits implementing the routing gadgets is $\mathcal{O}(mn)$. These two terms dominate the overall size of the circuit. The depth of the circuit is dominated by the depth of the $\cktRoute{n}$.
\end{proof}

We can use the above circuit in an iterative fashion to build a smaller circuit for the same primitive.

\begin{lemma}
	There is a family of boolean circuits $\cktRSort{n}{m}{1} \fromto{n \cdot m}{n \cdot m}$
	that on input $x_1, x_2, \ldots, x_n \in \left\{ 0,1 \right\}^{m}$ partially sort these numbers according to their first most significant bit.
	The size of the circuit \cktRSort{n}{m}{1} is $\mathcal{O}(n m (1+\log^*(n)-\log^*(m)))$ and its depth is $\mathcal{O}(\log^2(n))$.
	\label{lem:route}
\end{lemma}

\begin{proof}
Assume $m\le \log(n)/11$ otherwise use Corollary~\ref{col:pippenger_sort}.
We will build the circuit iteratively using the circuit from Corollary~\ref{col:pippenger_sort} for blocks of various sizes. We will start with small blocks of items and we will iteratively sort larger and larger number of items organized into mostly monochromatic blocks. 
Without loss of generality we assume that $m$ is a power of two, and we will ignore the rounding issues.
We will have two parameters $m_i$ and $n_i=2^{3m_i}$, where $m_0=m$ and $m_{i+1}=2^{m_i}$ for $i\ge 0$. At iteration $i$, all the items will be partitioned into \emph{parts} of consecutive numbers, each part will be either \emph{monochromatic} containing all zeros, all ones, or it will be \emph{mixed}. (Here we refer to the most significant bits of the numbers in the part.) For each part we will maintain two indicator bits which of the three possibilities occurs: an indicator which is one if the block is mixed, and another \emph{color} indicator which specifies the highest order bit of the integers if the block is monochromatic. (For the latter we could use the first bit of the first integer in the part.)
At each iteration $i>0$, $m_i$ will denote the number of items in each part. $n_i/m_i$ consecutive parts form a \emph{block}, so each block contains $n_i$ items. The blocks partition the input. We will maintain an invariant that the fraction of mixed parts in each block is at most $2/m^3_i$.

At iteration $0$ we apply $\cktPSort{n_0}{m}{1}$ to consecutive blocks of $n_0$ input integers. Afterwards, the block is partitioned into parts of size $m_1$ and for each part we determine its status by comparing the most significant bits of the first and last integer in the part. It is clear that each block of size $n_0$ contains at most one mixed part. As the number of parts in the block is $m_1^3$, the fraction of mixed parts in each block is at most $2/m^3_1$, and this is also true for blocks of size $n_1$.

At iteration $i>0$, we divide the current sequence of parts of size $m_i$ into blocks containing $n_i/m_i$ parts, and we proceed in three steps:

\begin{description}

	\item[Step 1.]
		Sort the parts in each block using $\cktPSort{n_i/m_i}{2+m_i\cdot m}{1}$ according to the mixed indicator.
		Hence, all the mixed parts will move to the end of the block.
		There are at most $2n_i/m^3_i$ mixed parts in each block, the remaining parts must be monochromatic.

	\item[Step 2.] 
		In each block, sort all the $m$-bit integers in the last $2n_i/m^3_i$ parts according to their most significant bit using $\cktPSort{2n_i/m^2_i}{m}{1}$.
		This sorts together all the integers in the mixed parts (and perhaps few other parts).
		Repartition them into parts of $m_i$ consecutive numbers and determine their indicator bits.
		Only one of the parts should be mixed at this point.
		Swap it with the last part in the block.
		(We provide details of the swap later.)

	\item[Step 3.] 
		In each block, sort all the parts except for the last one according to their color indicator using $\cktPSort{(n_i/m_i)-1}{2+m_i\cdot m}{1}$.
		This moves all the parts of color 0 to the front.
		Repartition all the numbers in the block into parts of $m_{i+1}$ consecutive integers and determine their indicator bits, where the last part is marked as mixed.
		At most two of the new parts should be mixed at this point.
		Notice, that out of $m_{i+1}^3$ parts in each block, at most two are marked as mixed so the invariant applies.
		We can move to the next iteration.

\end{description}

We iterate the algorithm until $m_i \ge \log(n)/4$. Once $m_i\ge \log(n)/4$, the number of integers in mixed parts is at most $2n/m^2_i \leq O(n/\log^2(n))$, remaining items are in monochromatic parts. At this point we cannot form a block of size $n_i$, but we can still perform the same type of actions as in Steps 1-3: We can bring the monochromatic parts forward as in Step~1, sort the last $32n/\log^2(n)$ integers belonging to the mixed parts, move the remaining mixed part to the end, sort the monochromatic parts and swap the mixed part with the first monochromatic part of color 1.

To swap a single mixed part with the last part we can copy the mixed part into a buffer by AND-ing every part bit-wise with the indicator whether that is the mixed part, and OR-ing all the results together. This copies the mixed part into a buffer. In a similar fashion we can copy the last part into the now unused part by letting each part bit-wise copy to its place either its original content or the content of the last part, again conditioning on an appropriate indicator bit. Hence, the swap can be implemented by a circuit of size proportional to the total size of the parts and depth logarithmic in the number of parts.

Now we will bound the total size of the circuit we constructed. Step~1 requires $n/n_i$ circuits of size $\mathcal{O}(n_i m + n_i/m_i \log(n_i/m_i))=
\mathcal{O}(n_i m)$, as $\log(n_i)=O(m_i)$, and of depth at most $\mathcal{O}(\log^2(n_i))$. Step~2 requires $n/n_i$ sorting circuits of size $\mathcal{O}(m n_i/m^2_i + 2n_i/m^2_i \log(2n_i/m^2_i))=\mathcal{O}(n_i)$ and of depth at most $\mathcal{O}(\log^2(n_i))$, together with a circuit of total linear size $\mathcal{O}(n)$ to recalculate the parts and do the swaps.  The last step requires the same amount of circuitry as the first step.

Hence, each step requires circuits of total size $\mathcal{O}(n m)$.
The same goes for the initial sort at iteration 0, and the final sorts at the end.
As there are at most $\log^*(n)-\log^*(m)$ iterations, the resulting size is $\mathcal{O}(n m (\log^*(n)-\log^*(m)))$.
Each step requires a circuit of depth $\mathcal{O}(\log^2(n_i))$, recall that by our choice $n_i = 2^{4m_i}$, thus $\log(n_{i}) = 4m_i$.
Since $m_{i+1} = 2^{m_i}$ and for each $i$ we have $m_i \leq \log(n) / 4$, thus the total depth is dominated by the last iteration where we use a circuit of depth $\mathcal{O}(\log^2(n))$.
\end{proof}

\begin{proof}[Proof of Theorem~\ref{thm:sort_by_k_bits}]
We assume that $k\le \log(n)/11$ otherwise we can use Corollary~\ref{col:aks_partial_circuit} to sort the elements. Without loss of generality we assume $n$ is a power of two.
We think of the input as organized into an array. We extract the first $k$ bits (\emph{key}) from each input element and we sort the keys using
the circuit from Theorem~\ref{thm:main} of size $\mathcal{O}(nk^2)$ and depth  $\mathcal{O}(\log(n) + k \log(k))$.

We will build recursively a circuit that will sort the input array of $n$ elements according to the first $k$ bits when the input is augmented with the array of sorted keys. Now our goal is to split the input array into two equal sized parts $L$ and $R$ where all elements in $L$ are less or equal to elements in $R$ when comparing only the keys. 

To do that we take the \emph{median}, the $n/2$-th element among the keys, and we partition the array according to it. We split the input array into three arrays $L$, $M$, and $R$ of length $n$ with elements less than, equal to, and greater than the median, resp., and we mark the unused elements as \emph{dummy} using an extra bit associated to each element. We sort $L$ and $M$ so that all non-dummy elements are to the left and $R$ so that all non-dummy elements are to the right. We use three circuits $\cktRSort{n}{m+1}{1}$ to do that. Now, we flip the first half of elements in $M$, i.e., swap the $i$-th element with the element in position $(n/2)-i+1$, and we replace the dummy elements in the first half of $L$ by the corresponding elements in $M$. By one application of $\cktRSort{n}{m+1}{1}$ we move all the remaining non-dummy elements in $M$ to the left,
and we merge those elements with the second half of $R$. We discard the second and first half of $L$ and $R$, respectively. (They contain only dummy elements.) 

If the highest order bit of the median is set to $0$ then all the elements in $L$ have the highest order bit set to $0$, otherwise all the elements in $R$ have the highest order bit set to $1$. In either case we reduced the problem to one problem of sorting half of the elements according to $k-1$ bits and the other half according to $k$-bits. We recursively build a circuit to sort $\cktSort{n/2}{m}{k-1}$ and  $\cktSort{n/2}{m}{k}$
when the input is augmented with the sorted array of keys. We pass to each of the sorting sub-circuits the appropriate sub-problem and we re-route the results from them to form the final output.

Not counting the two sub-circuits  $\cktSort{n/2}{m}{k-1}$ and  $\cktSort{n/2}{m}{k}$, this step requires four copies of the circuit $\cktRSort{n}{m+1}{1}$ and additional $\mathcal{O}(nm)$ gates to do the moves and element comparison with the median. Denote the size of this part of the circuit by $L_{m}(n)=\mathcal{O}(n m (1+\log^*(n)-\log^*(m)))$.
The depth of the resulting circuit to perform all those operations is $\mathcal{O}(\log^2(n))$ as the move operations are done in parallel (again, not counting the depth of $\cktSort{n/2}{m}{k-1}$ and  $\cktSort{n/2}{m}{k}$).
	If we denote by $S_{m, k}(n)$ the size of the circuit \cktSort{n}{m}{k} we get the following recurrence:
	\begin{align*}
		S_{m, k}(1) &= \mathcal{O}(m) \\
		S_{m, 1}(n) &= \mathcal{O}(n m (1 + \log^*(n) - \log^*(m))) \\
		S_{m, k}(n) 
		&\leq L_{m}(n) + S_{m, k-1}\left(  \frac{n}{2} \right) + S_{m, k}\left(  \frac{n}{2}  \right)
	\end{align*}
	when we iterate the recurrence:
	\begin{align*}
		S_{m, k}(n) &= L_{m}(n) + S_{m, k-1}(n/2) + S_{m, k}(n/2) \\
		&= L_{m}(n) + S_{m, k-1}(n/2) + L_{m}(n/2) + S_{m, k-1}(n/4) + S_{m, k}(n/4) \\
		&= L_{m}(n) + S_{m, k-1}(n/2) + L_{m}(n/2) \\
		&{} \text{ \  \  \    }  + S_{m, k-1}(n/4) + L_{m}(n/4) + S_{m, k-1}(n/8) + S_{m, k}(n/8) \\
		&= \ldots \\
		&= \left( L_{m}(n) + L_{m}(n/2) + \ldots + L_{m}(1) \right) \\
		&{} \text{\ \ \ \ } + \left( S_{m, k-1}(n/2) + S_{m, k-1}(n/4) + \ldots + S_{m, k-1}(1) \right) + S_{m, k}(1) \\
		&\leq L_{m}(2 n) + S_{m, k-1}(n) + \mathcal{O}(m)
	\end{align*}
	which gives us
	\begin{align*}
		S_{m, k}(n) &= k L_{m}(2 n) + (k-1) S_{m, k}(1) + S_{m, 1}(n) \\
		&= k L_{m}(2 n) + \mathcal{O}(n m (1 + \log^*(n) - \log^*(m))) \\
		&= \mathcal{O}(k n m (1 + \log^*(n) - \log^*(m)))
	\end{align*}

	To bound the depth $D_{m, k}(n)$ we use the following recurrence:
	\begin{align*}
		D_{m, k}(1) &= \mathcal{O}(1) \\
		D_{m, k}(n) &\geq D_{m, k-1}(n) \\
		D_{m, 1}(n) &= \mathcal{O}(\log^2(n)) \\
		D_{m, k}(n) &= \mathcal{O}(\log^2(n)) + \max\left( D_{m, k}(n/2) + D_{m, k-1}(n/2) \right) \\
		&\leq \mathcal{O}(\log^2(n)) + D_{m, k}(n/2) \\
		&\leq \mathcal{O}(\log^3(n))
	\end{align*}
\end{proof}

\section{Conclusion}

We have provided improved sorting circuits.
Our technique used in the proof of Theorem~\ref{thm:main} can be viewed as information compression and decompression.
This technique might prove useful for other related problems. We list some open problems:

\begin{itemize}

	\item  Most of our circuits are uniform.
		The non-uniform part is due to the use of the AKS circuits and Pippenger's super-concentrators.
		Can one make uniform circuits of the same size?

	\item 
		Kospanov~\cite{kospanov1994scheme} shows that there is a family of sorting circuits with depth $\mathcal{O}(\log(n) + \log(m))$ and size $\mathcal{O}(m n^2)$ that sorts $n$ numbers each of $m$ bits.
		Is there a circuit family for sorting with circuits of depth $\mathcal{O}(\log(n) + \log(m))$ and size $\mathcal{O}(n m^2)$?
		In other words can we get rid of the $m \log(m)$ factor in the circuit depth from Theorem~\ref{thm:main} while keeping the $\mathcal{O}(nm^2)$ size?

	\item  Is it possible to partially sort $n$ numbers of $m$ bits each by their first bit using a circuit of size $\mathcal{O}(nm)$ and depth $\mathcal{O}(\log(n))$?

\end{itemize}

\paragraph{Acknowledgement:} 
The authors are grateful for insightful discussions with Mike Saks on sorting and to Veronika Slívová for her insights and comments regarding the first versions of this paper.
The authors thank Igor Sergeev for pointing us to the paper of Kospanov~\cite{kospanov1994scheme}.

\bibliography{bibliography}

\end{document}